\documentclass[letterpaper, 10 pt, conference]{ieeeconf}  

\IEEEoverridecommandlockouts         

\usepackage{balance} 

\usepackage{xcolor,calc}

\usepackage[pdftex]{graphicx}
\usepackage{amsmath} 
\usepackage{amssymb}  
\usepackage[noadjust]{cite}
\usepackage{color}
\usepackage{enumerate}
\usepackage{multirow}
\usepackage{mathtools}
\usepackage{booktabs}
\usepackage{epstopdf}
\usepackage{algcompatible}
\usepackage[linesnumbered,ruled,vlined]{algorithm2e}
\usepackage[normalem]{ulem}

\usepackage{url}
\usepackage{bm}
\usepackage{caption}

\usepackage{booktabs}

\newtheorem{theorem}{Theorem}
\newtheorem{remark}{Remark}
\newtheorem{proposition}{Proposition}

\newtheorem{lemma}{Lemma}
\newtheorem{corollary}{Corollary}
\SetKwInput{KwInit}{Initialization}                
\SetKwInput{KwInput}{Input}

\makeatletter
\newsavebox\myboxA
\newsavebox\myboxB
\newlength\mylenA

\newcommand*\xoverline[2][0.75]{%
    \sbox{\myboxA}{$\m@th#2$}%
    \setbox\myboxB\null
    \ht\myboxB=\ht\myboxA%
    \dp\myboxB=\dp\myboxA%
    \wd\myboxB=#1\wd\myboxA
    \sbox\myboxB{$\m@th\overline{\copy\myboxB}$}
    \setlength\mylenA{\the\wd\myboxA}
    \addtolength\mylenA{-\the\wd\myboxB}%
    \ifdim\wd\myboxB<\wd\myboxA%
       \rlap{\hskip 0.5\mylenA\usebox\myboxB}{\usebox\myboxA}%
    \else
        \hskip -0.5\mylenA\rlap{\usebox\myboxA}{\hskip 0.5\mylenA\usebox\myboxB}%
    \fi}
\makeatother

\newcommand{\R}{\mathbb{R}}

\usepackage{tabularx}
\usepackage{booktabs}

\newtheorem{definition}{Definition}

\makeatletter
\let\NAT@parse\undefined
\makeatother
\usepackage{hyperref}

\title{\LARGE \bf A Semi-Definite Programming Approach to Robust Adaptive MPC under State Dependent Uncertainty}

\author{Monimoy Bujarbaruah$^\star$, Siddharth H. Nair$^\star$, and Francesco Borrelli  
\thanks{The authors are with the MPC Lab, UC Berkeley, USA; E-mails: \tt\scriptsize{\{monimoyb,siddharth\_nair,fborrelli\}@berkeley.edu.}}
\thanks{$^\star$ These authors contributed equally to this work.
}
}

\begin{document}

\maketitle
  \thispagestyle{empty}
\pagestyle{empty}
\begin{abstract}
We propose an Adaptive MPC framework for uncertain linear systems to achieve robust satisfaction of state and input constraints. The uncertainty in the system is assumed additive, state dependent, and globally Lipschitz with a known Lipschitz constant. We use a non-parametric technique for online identification of the system uncertainty by approximating its \emph{graph} via envelopes defined by quadratic constraints. At any given time, by solving a set of \emph{convex} optimization problems, the MPC controller
guarantees robust constraint satisfaction for the closed-loop system for all possible values of system uncertainty modeled by the envelope. The uncertainty envelope is refined with data using Set Membership Methods. We highlight the efficacy of the proposed framework via a detailed numerical example. 
\end{abstract}

\section{Introduction}
System modeling and identification has been an integral part of statistics and data sciences \cite{friedman2001elements, hastie2017generalized}. In recent times, as data-driven decision making and control becomes ubiquitous \cite{recht2018tour, rosolia2018data}, such system identification methods are being integrated with control algorithms for control of uncertain dynamical systems. In computer science community, data driven reinforcement learning algorithms \cite{bertsekas1989adaptive, watkins1992q} have been extensively utilized for policy and value function learning of uncertain systems. In control theory, if the actual model of a system is unknown, adaptive control \cite{krstic1995nonlinear, sastry2011adaptive} strategies have been applied for simultaneous system identification and control. In such classical adaptive control methods, primarily unconstrained systems are considered, and model parameters are learned from data in terms of point estimates, while proving stability of the closed-loop system. 

The concept of online model learning and adaptation has been extended to control design for systems under constraints as well. 
In \cite{dean2018safely}, linear time invariant system dynamics matrices and the confidence intervals are learned using Ordinary Least Squares regression and imposed constraints are robustly satisfied using System Level Synthesis \cite{doyle2017system}. Lowering the conservatism of such an approach, the field of Adaptive MPC has gained attention in recent times \cite{tanaskovic2014adaptive, bujarbaruahAdapFIR, lorenzenAutomaticaAMPC, Khler2019LinearRA, bujarArxivAdap}. In the aforementioned Adaptive MPC frameworks, Set Membership Method based approaches are used to obtain the sets containing all possible realizations of model uncertainty. These sets are then modified as more data becomes available. However, the model uncertainty learned is not considered as a function of system states. Work such as \cite{hewing2017cautious, soloperto2018learning, koller2018learning} extend the Adaptive MPC framework to systems with state dependent uncertainties, where set based bounds of the uncertainty is adapted using Gaussian Process (GP) regression. However, due to probabilistic nature of GP regression based estimates, there is no closed-loop constraint satisfaction guarantees in such methods. 
To the best of our knowledge, there has not been a unifying framework for robust MPC design in presence of state dependent system model uncertainty. 

In this paper, we propose a novel approach to designing a robust Adaptive MPC algorithm for linear systems subject to state and input constraints, in presence of state dependent system uncertainty. The uncertainty is assumed globally Lipschitz, with a known Lipschitz constant. We utilize a non-parametric recursive system identification strategy \cite{nair_LarX}, which identifies the \emph{graph} of the uncertainty from data using its Lipschitz property. The identification is successively refined with recorded data. Our main  contributions are: 
\begin{itemize}
    \item We provide set based bounds containing all possible realizations of the system uncertainty, using its Lipschitz property. This in contrast to the probabilistic nature of bounds in \cite{hewing2017cautious, soloperto2018learning, koller2018learning}, due to the use of GP regression. Our uncertainty set bounds are modified successively with set intersections upon gathering new data. 
    
    \item Utilizing the above bounds on system uncertainty, we synthesize a robust Adaptive MPC controller by solving convex optimization problems, satisfying imposed state and input constraints. We prove its recursive feasibility, extending feasibility guarantees of \cite{tanaskovic2014adaptive, lorenzenAutomaticaAMPC, bujarArxivAdap} in presence of state dependent uncertainty.  We further demonstrate the validity and efficacy of the proposed approach through a detailed numerical simulation.  
\end{itemize}
The paper is organized as follows: Section~\ref{sec:prob_for} formulates the robust optimization problem to be solved for the uncertain system, along with the system model and constraints. Section~\ref{sec:unc_adap} contains the recursive system uncertainty adaptation framework. The tractable robust Adaptive MPC is posed in Section~\ref{sec:data_dmpc}. We present numerical results in Section~\ref{sec:numerics}. 

\section{Problem Formulation}\label{sec:prob_for}
\subsection{System Model}
The system is given by:
\begin{align}\label{eq:unc_sys}
x_{t+1} = A x_t + B u_t + d(x_t), 
\end{align}
where $x_t\in \mathbb{R}^{n}$ is the state at time $t$, $u_t\in\mathbb{R}^{m}$ is the input, $A$ and $B$ are known system matrices of appropriate dimensions, and $d(x_t)$ constitutes un-modelled dynamics, that is, the system uncertainty, which is $L_d$ Lipschitz in its convex and closed domain $\mathrm{dom}(d)$ with a known $L_d$. 

\subsection{Constraints}
The system dynamics are subject to polytopic state and input constraints of the form: 
\begin{subequations}\label{eq:constraints}
\begin{align}
\mathcal{X}&=\{x\in\mathbb{R}^n\ \vert H_x x \leq h_x \}, \label{eq:state_con}\\
\mathcal{U}&=\{u\in\mathbb{R}^m\ \vert H_uu\leq h_u \} \label{eq:input_con},
\end{align}
\end{subequations}
where we assume $\mathcal{X} \subseteq \mathrm{dom}(d)$. 

\subsection{Robust Optimization Problem}
Our goal is to design a controller that solves the following infinite horizon optimal control problem with constraints \eqref{eq:constraints}
\begin{equation}\label{eq:OP_inf}
	\begin{array}{clll}
		\hspace{0cm} 
	\displaystyle\min_{u_0,u_1(\cdot),\ldots} & \displaystyle\sum\limits_{t\geq0} \bar{x}_t^\top Q \bar{x}_t + u_t^\top(\bar{x}_t) R u_t(\bar{x}_t)  \\[1ex]
		\text{s.t.}  & x_{t+1} = Ax_t + Bu_t(x_t) + d(x_t),\\
		& \bar{x}_{t+1} = A \bar{x}_t + Bu_t(\bar{x}_t) + \bar{d}(\bar{x}_t),\\[1ex]
		& H_x x_{t+1} \leq h_x,\ ~\forall d(x_{t}) \in \mathcal{D}(x_{t}),\ \\
		& H_u u(x_{t}) \leq h_u,\ \forall d(x_{t}) \in \mathcal{D}(x_{t}),\ \\[1ex]
		&  x_0 = \bar{x}_0 = x_S,\\ 
		& t=0,1,\ldots,
	\end{array}
\end{equation}
where $\mathcal{D}(x_t)$ is a state dependent compact set where the uncertainty $d(x_t)$ is guaranteed to lie, and $\bar{d}(\bar{x}_t)$ denotes the certainty equivalent (nominal) estimate of uncertainty at any point $\bar{x}_t$ along the nominal trajectory.  Matrices $Q,R \succ 0$ are weight matrices. We point out that, as system \eqref{eq:unc_sys} is uncertain, the optimal control problem \eqref{eq:OP_inf} consists of finding input policies $[u_0,u_1(\cdot),u_2(\cdot),\ldots]$, where $u_t: \mathbb{R}^{n}\ni x_t \mapsto u_t = u_t(x_t)\in \mathbb{R}^m$. We wish to approximate solutions to optimization problem \eqref{eq:OP_inf} by solving the following finite time constrained optimal control problem at each time $t$, in a receding horizon fashion:
\begin{equation} \label{eq:mpc_prob_tough}
    \begin{array}{llll}
        \displaystyle \min_{u_{t|t}, u_{t+1|t}(\cdot),\dots} & \mathlarger{\sum}_{k=t}^{t+N-1} (\bar{x}_{k|t}^\top Q \bar{x}_{k|t} + u_{k|t}^\top(\bar{x}_{k|t}) Ru_{k|t}(\bar{x}_{k|t})) 
        \\ & ~~~~~~~~~~~~~~~~~~~~~~~~~~+ \bar{x}_{t+N|t}^\top P_N \bar{x}_{t+N|t}  
        \vspace{2mm} \\ 
        \ \ \quad  \text{s.t. }  & 
         {x}_{k+1|t} = A{x}_{k|t} + Bu_{k|t}(x_{k|t}) + d(x_{k|t}),\\
         & \bar{x}_{k+1|t} = A\bar{x}_{k|t} + Bu_{k|t}(\bar{x}_{k|t}) + \bar{d}(\bar{x}_{k|t}),\vspace{.6em}\\ 
        & H_x x_{k+1|t} \leq h_x,~\forall d(x_{k|t}) \in \mathcal{D}(x_{k|t}),\\
        & H_u u_{k|t}(x_{k|t}) \leq h_u,~\forall d(x_{k|t}) \in \mathcal{D}(x_{k|t}), \\ 
        & \forall k = t,t+1,\dots t+N-1,\\
        & x_{t+N|t} \in \mathcal{X}_N,\\
        & x_{t|t} = \bar{x}_{t|t} = x_t,
    \end{array}
\end{equation}
where $x_{k|t}$ is the predicted state after applying the predicted policy $[u_{t|t}(x_{t|t}),\dots,u_{k-1|t}(x_{k-1|t})]$ for $k=\{t+1,\dots,t+N\}$ to system \eqref{eq:unc_sys}, $\mathcal{X}_N$ is the terminal set and $P_N \succ 0$ is the terminal cost. In the following sections, we address the three crucial challenges associated to finding solutions of \eqref{eq:mpc_prob_tough}: 
\begin{enumerate}[i)]
    \item Learning and updating the uncertainty bounds $\mathcal{D}(\cdot)$ with data to obtain a nonempty $\mathcal{X}_N$.
    \item Obtaining tractable parametrization of input policy $u(\cdot)$ to avoid searching over infinite dimensional function spaces, and 
    \item Ensuring robust satisfaction of \eqref{eq:constraints} for all times, if tractable reformulation of \eqref{eq:mpc_prob_tough} is feasible once.
\end{enumerate}

\section{Uncertainty Learning and Adaptation}\label{sec:unc_adap}
At every time instant $t$, we assume that we have access to measurements $d(x_{i})$ for all $i=\{0,1,\dots,t-1\}$, that is, the realizations of the uncertainty function. 
    

\subsection{Successive Graph Approximation}
\begin{definition}[Graph]\label{def:graph}
The graph of a function $f: \mathbb{R}^n\rightarrow \mathbb{R}^n$ is defined as the set 
\begin{align*}
    G(f)=\{(x,f(x))\in\mathbb{R}^n\times\mathbb{R}^n\vert\ \forall x\in\mathrm{dom}(f)\}.
\end{align*}
\end{definition}
We use quadratic constraints (QCs) as our main tool to approximate the graph of a function. A definition appropriate for our purposes is presented below. 

\begin{definition}[QC Satisfaction] A set $\mathcal{A}\subset \mathbb{R}^{2n}$ is said to satisfy the quadratic constraint specified by symmetric matrix $Q_c$ if
\begin{equation*}
    \begin{bmatrix}x\\ 1\end{bmatrix}^\top Q_c\begin{bmatrix}x\\ 1\end{bmatrix}\leq 0, \quad \forall x\in\mathcal{A}.
\end{equation*}
\end{definition}
The following proposition uses a QC to characterise a coarse approximation of the graph of an $L_d-$Lipschitz function.
\begin{proposition}\label{prop:qc}
The graph $G(d)$ of the $L_d-$Lipschitz function $d(\cdot)$ inferred  at any time $t$, using the measurement $(x_i,d(x_i))$ for any $0\leq i < t$, satisfies the QC specified by the matrix
\small
\begin{align*}
Q_L^d(x_i)=\begin{bmatrix}
-L_d^2\mathbb{I}_n& \mathbf{0}_{n\times n} & L_d^2x_i\\
\mathbf{0}_{n\times n} & \mathbb{I}_n & -d(x_i)\\
L_d^2 (x_i)^\top & -d^\top(x_i) & -L_d^2  (x_i)^\top x_i \\
& & \hspace{8mm} + d^\top(x_i) d(x_i)
\end{bmatrix},
\end{align*}
where $\mathbb{I}_n$ denotes the identity matrix of size $n \times n$ and $d(x_i) = x_{i+1}-Ax_i - Bu_i(x_i)$. 
\end{proposition}

\begin{proof}
Since $d(\cdot)$ is $L_d-$Lipschitz, we have by definition for $(x_t, d(x_t))\in G(d)$ at any time $t$, and $(x_i, d(x_i))$ measured at any $i < t$
\small
\begin{align*}
    &  \Vert (d(x_t)-d(x_i)) \Vert ^2 \leq L_d^2 \Vert (x_t-x_i) \Vert^2, \\
    \iff &\begin{bmatrix} x_t \\ d(x_t) \\ 1 \end{bmatrix}^\top Q_L^d(x_i)\begin{bmatrix}x_t\\d(x_t)\\1 \end{bmatrix}\leq 0, \quad \forall (x_t,d(x_t))\in G(d).
\end{align*}
\end{proof}
\begin{definition}[Envelope]\label{def:envelope}
An envelope of a function $f:\mathbb{R}^n\rightarrow\mathbb{R}^n$ is defined as any set $\mathbf{E}^f\subseteq \mathbb{R}^n\times\mathbb{R}^n$ with the property
\begin{align*}
    G(f)\subseteq \mathbf{E}^f.
\end{align*}
\end{definition}
\begin{corollary}\label{corr:env}
The set defined by 
\begin{align*}
\mathcal{E}(x_i)=\{(x,d)\in\mathbb{R}^{2n} : \begin{bmatrix}x\\d\\1 \end{bmatrix}^\top Q_L^d(x_i)\begin{bmatrix}x\\d\\1 \end{bmatrix}\leq 0\}
\end{align*}
is an \textit{envelope} containing the graph of  $L_d-$Lipschitz function $d(\cdot)$ for all times $t \geq 0$, after collecting measurements $(x_i, d(x_i))$ for any $i = 0,1,\dots,t-1$.
\end{corollary}

\begin{lemma}\label{lem:inter}
Given a sequence of measurements $\{x_i\}_{i=0}^{t-1}$ obtained under dynamics \eqref{eq:unc_sys}, we have 
\begin{align}\label{eq:intenv}
    G(d)\subseteq\bigcap_{i=0}^{t-1} \mathcal{E}(x_i) = \mathbf{E}_t^d.
\end{align}
\end{lemma}
\begin{proof}
See \cite[Lemma~1]{nair_LarX}. 
\end{proof}

\subsection{Uncertainty Estimation at a Given State}
\label{ssec:estim_function}
We wish to obtain a set where the possible realizations of $d(x_t)$ can lie, which we denote by $\mathcal{D}(x_t)$, for any $x_t \in \mathcal{X}$. Using the collected tuple $(x_i, d(x_i))$ from any time instant $i<t$, we can obtain a set based estimate of the range of possible values of $d(x_t)$, called the \emph{sampled range set} as, 
\begin{align*}
    \mathcal{S}(x_i, x_t) & :=  \mathcal{E}(x_i) \Big \vert_{x = x_t} = \{d :  \begin{bmatrix} x_t \\ d \\ 1 \end{bmatrix}^\top Q^d_L(x_i) \begin{bmatrix} x_t \\ d \\ 1 \end{bmatrix}  \leq 0\},
\end{align*}
for any $i < t$. 
As we successively collect $(x_i,d(x_i))$ for $i=\{0,1,\dots, t-1\}$, the set of possible values of $d(x_t)$ is obtained and refined with intersection operations as
\begin{align}\label{eq:func_dom}
    \mathcal{D}(x_t) = \bigcap_{i=0}^{t-1} \mathcal{S}(x_i, x_t)=\bigcap_{i=0}^{t-1} \mathcal{E}(x_i) \Big \vert_{x = x_t},
\end{align}
with the guarantee $d(x_t) \in \mathcal{D}(x_t)$ at any given time $t \geq 0$. We further note that the set $\mathcal{D}(x_t)$ is convex, as it is an intersection of convex sets \cite{nair_LarX}. 

\begin{proposition}\label{prop:containment}
Consider a specific state $\tilde{x}$, at time instants $t_1$ and $t_2$, with $t_1 <t_2$. Denote them by $\tilde{x}_{t_1}$ and $\tilde{x}_{t_2}$ respectively. Then we have $\mathcal{D}(\tilde{x}_{t_2}) \subseteq \mathcal{D}(\tilde{x}_{t_1})$. 
\end{proposition}
\begin{proof}
See Appendix.
\end{proof}
\section{Robust Adaptive MPC Formulation}\label{sec:data_dmpc}
The main challenges addressed in this section are:
\begin{enumerate}
    \item Generalizing \eqref{eq:func_dom} to obtain set based uncertainty bounds along the prediction horizon of the MPC problem \eqref{eq:mpc_prob_tough},
    \item Posing a tractable robust optimization problem to solve \eqref{eq:mpc_prob_tough} with feasibility guarantees. 
\end{enumerate}

\subsection{Uncertainty Sets Along the MPC Horizon}\label{ssec:bounding_unc_set}
\begin{definition}{Robust Controllable States:} The \emph{1-Step Robust Controllable States} from any set $\mathcal{A}$ is defined as
\begin{align*}
    \mathrm{Succ}(\mathcal{A},\mathcal{W})  := & \{ x^+ \in \mathcal{X}: \exists x \in \mathcal{A}, \exists u \in \mathcal{U}, \exists w \in \mathcal{W},\\& ~~~~~~~~~~~~~~~\mathrm{s.t. }~ x^+ = Ax + Bu + w \}, 
\end{align*}
with state constraints $\mathcal{X}$ defined in \eqref{eq:state_con}.
\end{definition}
Given any state $x_t$, an s-procedure based approach to obtain an ellipsoidal outer approximation to $\mathcal{D}(x_t)$, denoted by $E^d(x_t)$, is presented in \cite[Section~V-A]{nair_LarX}. We then successively obtain ellipsoidal outer approximations for uncertainty sets $\mathcal{D}(\mathcal{X}_{k|t})$, that is, $E^d(\mathcal{X}_{k|t}) \supseteq  \mathcal{D}(\mathcal{X}_{k|t})$, with 
\begin{align*}
    \mathcal{D}(\mathcal{X}_{k|t}) = \bigcup_{x_{k|t} \in \mathcal{X}_{k|t}}~ \mathcal{D}(x_{k|t}),
\end{align*}
where 
\begin{subequations}\label{eq:pred_succ}
\begin{align}
    & \mathcal{X}_{k|t} \supseteq \mathrm{Succ}(\mathcal{X}_{k-1|t},E^{d}(\mathcal{X}_{k-1|t})),\label{eq:succ_hor}\\
    & \forall k = t+1,t+2,\dots,t+N, \nonumber \\
    &\mathcal{X}_{t|t} = x_t,~\mathcal{X}_{t+N|t} = \mathcal{X}.\label{eq:X_last}
\end{align}
\end{subequations}
Let sets $E^d(\mathcal{X}_{k|t})$ for any $k=\{t,t+1,\dots,t+N\}$ be 
\begin{align}\label{eq:d_set_def}
    E^d(\mathcal{X}_{k|t}) & := \{d: (d-p^d_{k|t})^\top q^d_{k|t}(d-p^d_{k|t}) \leq 1\},\nonumber \\
    & := \begin{bmatrix}d \\1 \end{bmatrix}^\top  \bar{P}^d_{k|t}  \begin{bmatrix} d\\1 \end{bmatrix} \leq 0,
\end{align}
with $\bar{P}^d_{k|t} = \begin{bmatrix} q^d_{k|t}&-q^d_{k|t} p^d_{k|t} \\-(p^d_{k|t})^\top q^d_{k|t} & (p^d_{k|t})^\top q^d_{k|t}(p^d_{k|t})-1\end{bmatrix}$, and center $p^d_{k|t} \in \R^n$ and positive definite shape matrix $q^d_{k|t} \in \mathbb{S}^n_{++}$ are decision variables. We consider parametrizations of sets $\mathcal{X}_{k|t}$ as
\begin{align}\label{eq:X_setDef}
     {\mathcal{X}}_{k|t} & := \{ x \in \mathbb{R}^n: (x-p^x_{k|t})^\top q^x_{k|t} (x-p^x_{k|t}) \leq 1\}, \nonumber \\
    & := \begin{bmatrix}x \\1 \end{bmatrix}^\top  \bar{P}^x_{k|t}  \begin{bmatrix} x\\1 \end{bmatrix} \leq 0,
\end{align}
where $\bar{P}^x_{k|t} = \begin{bmatrix} q^x_{k|t}& -q^x_{k|t} p^x_{k|t}\\-(p^x_{k|t})^\top q^x_{k|t} & (p^x_{k|t})^\top q^x_{k|t} (p^x_{k|t})-1\end{bmatrix}$ for any $k=\{t,t+1,\dots,t+N\}$. Center $p^x_{k|t} \in \R^n$ and shape matrix $q^x_{k|t} \in \mathbb{S}^n_{++}$ can be successively chosen  satisfying \eqref{eq:succ_hor}, with $p^x_{t|t} = x_t$ and $q^x_{t|t} = \mathrm{diag}(\infty, \dots, \infty) \in \mathbb{S}_{++}^{n}$, if sets $E^d(\mathcal{X}_{k|t})$ are found.
\begin{proposition}\label{prop:containment2}
Using s-procedure, $E^d(\mathcal{X}_{k|t})$ is obtained if the following holds true for some scalars $\{\rho^k_t, \tau^k_0, \tau^k_1, \dots, \tau^k_{t-1}\} \geq 0$ at each $k = \{t,t+1,\dots,t+N\}$, for all times $t \geq 0$:
\begin{align}\label{eq:bigE}
&\begin{bmatrix} -\rho^k_t q^x_{k|t}& 0 & \rho^k_t q^x_{k|t}p^{x}_{k|t}\\0& q^d_{k|t} & -q^d_{k|t} p^{d}_{k|t} \\ \rho^k_t (p^{x}_{k|t})^\top q^x_{k|t} &-(p^{d}_{k|t})^\top q^d_{k|t} & (p^d_{k|t})^\top q^d_{k|t}(p^d_{k|t}) -1 \\
& & \hspace{2mm} + \rho^k_t  -\rho^k_t (p^{x}_{k|t})^\top q^x_t(p^{x}_{k|t})  \end{bmatrix}\nonumber\\ 
&
~~~~~~~~~~~~~~~~~~~~~~~~~~~~~~~~~~~~~~~~~~~ - \sum \limits_{i=0}^{t-1} \tau^k_i Q_L^d(x_i) \preceq 0.
\end{align}
\end{proposition}
\begin{proof}
See Appendix. 
\end{proof}
We reformulate the feasibility problem \eqref{eq:bigE} as a Semi-definite Program (SDP) in the Appendix. After finding $E^d(\mathcal{X}_{k|t})$ using \eqref{eq:bigE}, to efficiently compute \eqref{eq:succ_hor}, we use polytopic\footnote{choice of this polytope is designer specific} outer approximations  $\mathcal{P}^d(\mathcal{X}_{k|t}) \supseteq E^d(\mathcal{X}_{k|t})$ instead of $E^d(\mathcal{X}_{k|t})$, given by
\begin{align} \label{eq:uncer_predBound}
    & \mathcal{P}^d(\mathcal{X}_{k|t}) := \{d: H^d_{k|t}d \leq h^d_{k|t}\},\\
    & \forall k= \{t,t+1,\dots,t+N\} \nonumber. 
\end{align}

\begin{remark}
Consider the state $x_{k|t}$ for prediction step $k$ at time $t$ in \eqref{eq:mpc_prob_tough}. From Proposition~\ref{prop:containment2} we know that $d(x_{k|t}) \in \mathcal{D}(\mathcal{X}_{k|t}) \Rightarrow d(x_{k|t}) \in \mathcal{P}^d(\mathcal{X}_{k|t})$, but $d(x_{k|t}) \in \mathcal{P}^{d}(\mathcal{X}_{k|t}) \nRightarrow d(x_{k|t}) \in \mathcal{D}(\mathcal{X}_{k|t})$. As a consequence, $\mathcal{P}^d(\mathcal{X}_{k|t}) \nsubseteq \mathcal{P}^d(\mathcal{X}_{k|t-1})$ is  possible. Hence, for ensuring recursive feasibility of solved MPC problem (detailed in Theorem~\ref{thm1}), we impose constraints in \eqref{eq:mpc_prob_tough} robustly for all $d(x_{k|t})$ satisfying  
\begin{align}\label{eq:prev_unc}
    & d(x_{k|t}) \in \mathcal{P}^d(\mathcal{X}_{k|t}) \cap \mathcal{P}^d(\mathcal{X}_{k|t-1}),\\
    & \forall k \in \{t,\dots,t+N-1\} \nonumber,  
\end{align}
with the initialization $\{q^d_{-1|-1},q^d_{0|-1},\dots,q^d_{N-2|-1}\} = \{0_{n \times n},\dots,0_{n \times n}\} \in \R^{n \times Nn}$. 
\end{remark}

\subsection{Control Policy Parametrization}
We restrict ourselves to the affine disturbance feedback parametrization \cite{Goulart2006,lofberg2003minimax} for control synthesis in \eqref{eq:mpc_prob_tough}. For all $k \in \{t,\cdots,t+N-1\}$ over the MPC horizon (of length $N$), the control policy is given as:
\begin{equation}\label{eq:inputParam_DF_OL}
	u_{k|t}(x_{k|t}) = \sum \limits_{l=t}^{k-1}M_{k,l|t} d(x_{l|t})  + v_{k|t},
\end{equation}
where $M_{k|t}$ are the \emph{planned} feedback gains at time $t$ and $v_{k|t}$ are the auxiliary inputs. Let us define $\mathbf{d}(x_t) = [d(x_{t|t}),\cdots,d(x_{t+N-1|t})]^\top \in \mathbb{R}^{nN}$. Then the sequence of predicted inputs from \eqref{eq:inputParam_DF_OL} can be compactly written as $\mathbf{u}_t = \mathbf{M}_t\mathbf{d}(x_t) + \mathbf{v}_t$ at any time $t$, where $\mathbf{M}_t\in \mathbb{R}^{mN \times nN}$ and $\mathbf{v}_t \in \mathbb{R}^{mN}$ are
\begin{equation*}
\begin{aligned}
    \mathbf{M}_t & =  \begin{bmatrix}0& \cdots&\cdots&0\\
  M_{t+1,t}& 0 & \cdots & 0\\
  \vdots &\ddots& \ddots &\vdots\\
  M_{t+N-1,t}& \cdots& M_{t+N-1,t+N-2}& 0
  \end{bmatrix}, \\ 
   \mathbf{v}_t & = [v_{t|t}^\top, \cdots, \cdots, v_{t+N-1|t}^\top]^\top.
\end{aligned}
\end{equation*} 
\subsection{Terminal Conditions}
We use state feedback to construct terminal set $\mathcal{X}_N$, which is the maximal robust positive invariant set \cite{kolmanovsky1998theory} obtained with a state feedback controller $u=Kx$, dynamics \eqref{eq:unc_sys} and constraints \eqref{eq:constraints}. This set has the properties
\begin{equation}\label{eq:term_set_DF}
    \begin{aligned}
    &\mathcal{X}_N \subseteq \{x|H_x x \leq h_x,~H_uKx \leq h_u\},\\
    &(A+BK)x + d(x) \in \mathcal{X}_N,~\\
    &\forall x\in \mathcal{X}_N,~\forall d(x) \in \mathcal{P}^d(\mathcal{X}).
    \end{aligned}
\end{equation}
Fixed point iteration algorithms to numerically compute \eqref{eq:term_set_DF} can be found in \cite{borrelli2017predictive, kouvaritakis2016model}. 
\subsection{Tractable MPC Problem}
The tractable MPC optimization problem at time $t$ is given by:
\begin{equation} \label{eq:MPC_R_fin}
	\begin{aligned}
		& \min_{\mathbf{M}_t, \mathbf{v}_t} ~~ \sum_{k=t}^{t+N-1} (\bar{x}_{k|t}^\top Q \bar{x}_{k|t} +  v_{k|t}^\top R v_{k|t}) + \bar{x}_{t+N|t}^\top P_N \bar{x}_{t+N|t} \\
		& ~~~\text{s.t}~~~~~~    x_{k+1|t} = Ax_{k|t} + Bu_{k|t}(x_{k|t}) + d(x_{k|t}),\\
		& ~~~~~~~~~~~~\bar{x}_{k+1|t} = A\bar{x}_{k|t} + Bv_{k|t} + \bar{d}_{k|t},\\
		&~~~~~~~~~~~~u_{k|t}(x_{k|t}) = \sum \limits_{l=t}^{k-1}M_{k,l|t} d(x_{l|t})  + v_{k|t},\\
		&~~~~~~~~~~~~H_x x_{k+1|t} \leq h_x, \\
		&~~~~~~~~~~~~ H_u u_{k|t}(x_{k|t}) \leq h_u,\\
	    & ~~~~~~~~~~~~ \forall d(x_{k|t}) \in \mathcal{P}^d(\mathcal{X}_{k|t}) \cap \mathcal{P}^d(\mathcal{X}_{k|t-1}),\\
	    & ~~~~~~~~~~~~ \forall k = \{t,\ldots,t+N-1\},\\
        	&~~~~~~~~~~~~x_{t+N|t} \in \mathcal{X}_N,~d(x_{N|t}) \in \mathcal{P}^d(\mathcal{X}),\\
        	&~~~~~~~~~~~~x_{t|t} = x_t, \ \bar{x}_{t|t} = x_t, \  \bar{d}_{k|t} \in \mathcal{P}^d(\mathcal{X}_{k|t}),
	\end{aligned}
\end{equation}
where $x_{k|t}$ is the predicted state after applying the predicted policy $[u_{t|t}(x_{t|t}),\dots,u_{k-1|t}(x_{k-1|t})]$ for $k=\{t+1,\dots,t+N\}$ to system \eqref{eq:unc_sys}, and the control invariant \cite{blanchini1999set} terminal set is $\mathcal{X}_N$. The parameters $\{p^d_{k|t}, q^d_{k|t}\}$ for $k=\{t,t+1,\dots,t+N\}$, that is, uncertainty containment ellipses in \eqref{eq:MPC_R_fin}, are computed before solving \eqref{eq:MPC_R_fin} at each time $t$, by finding solutions of \eqref{eq:bigE}. Nominal uncertainty estimate $\bar{d}_{k|t}$ is chosen as the Chebyshev center (i.e, center of the largest volume $\ell_2$ ball in a set) of $\mathcal{P}^d(\mathcal{X}_{k|t})$. After solving \eqref{eq:MPC_R_fin} at time $t$, in closed-loop we apply 
\begin{align}\label{eq:mpc_law}
    u_t(x_t) = v^\star_{t|t},
\end{align}
to system \eqref{eq:unc_sys} and then resolve \eqref{eq:MPC_R_fin} at $t+1$.

\begin{remark}
Terminal set $\mathcal{X}_N$ might be empty initially, due to conservatism resulting from a large volume of the set $\mathcal{P}^d(\mathcal{X})$. As more data is collected and the graph of $d(\cdot)$ is refined as in \eqref{eq:intenv}--\eqref{eq:func_dom}, $E^d(\mathcal{X})$, and so $\mathcal{P}^d(\mathcal{X})$ is refined with new data by solving \eqref{eq:bigE} (for only $k=t+N$, if data collected until instant $t$) with an updated $Q^d_L(\cdot)$. This eventually results in a nonempty $\mathcal{X}_N$. Once \eqref{eq:MPC_R_fin} is feasible with this $\mathcal{X}_N$, during the control process one may further update and enlarge $\mathcal{X}_N$ to lower conservatism of \eqref{eq:MPC_R_fin}. 
\end{remark}

\begin{algorithm}[h!]
    \caption{
 Robust Adaptive MPC
    }
    \label{alg1}
    \begin{algorithmic}[1]
      \Statex \hspace{-1.2em}\textbf{Initialize:} $\mathcal{P}^d(\mathcal{X}) = \R^n$; $j =0$; 
      \vspace{1.2mm}
      \Statex \hspace{-1.2em} \emph{begin exploration (offline)}

      \WHILE{$\mathcal{X}_N$ is empty}     
    
      \STATE Apply exploration inputs $u_j$ to \eqref{eq:unc_sys}. Collect
      
      \hspace{-5mm} $(x_j,d(x_j))$ at $j+1$. Set $j=j+1$; 
      
      \STATE Solve \eqref{eq:bigE} with $k=j+N$ to get $\mathcal{P}^d(\mathcal{X})$. 
      
      \hspace{-5mm} Compute $\mathcal{X}_N$ from \eqref{eq:term_set_DF};
      \ENDWHILE
      \Statex \hspace{-1.2em} \emph{end exploration} \textbf{set $j_{\mathrm{max}} \equiv t=0$}. 
      \vspace{1.2mm}
      
      \Statex \hspace{-1.2em} \emph{begin control process (online)} 
      \WHILE{during control for $t \geq 0$}
      
      \STATE Obtain $\mathcal{P}^d(\mathcal{X}_{k|t})$ for $k=\{t,t+1,\dots,t+N-1\}$ 
      
      \hspace{-5mm} from feasibility of  \eqref{eq:bigE};

      \IF{larger $\mathcal{X}_N$ desired}

       Update $\mathcal{P}^d(\mathcal{X})$ from \eqref{eq:bigE} (with $k=t+N$). 
       
       \hspace{-1mm} Update $\mathcal{X}_N$ from \eqref{eq:term_set_DF};
      \ENDIF      
      
      \STATE Solve \eqref{eq:MPC_R_fin} and apply MPC control \eqref{eq:mpc_law} to \eqref{eq:unc_sys};
      
      \ENDWHILE
      \end{algorithmic}
\end{algorithm}

\begin{theorem}\label{thm1}
Let optimization problem \eqref{eq:MPC_R_fin} be feasible at time $t=0$. Assume the state dependent uncertainty $d(\cdot)$ bounds along the horizon are obtained using \eqref{eq:bigE},~\eqref{eq:pred_succ}, and \eqref{eq:uncer_predBound}. Then,  \eqref{eq:MPC_R_fin} remains feasible at all times $t\geq 0$, if the state $x_t$ is  obtained by applying  the closed-loop MPC control law \eqref{eq:mpc_law} to system \eqref{eq:unc_sys}.
\end{theorem}

\begin{proof}
Let the optimization problem \eqref{eq:MPC_R_fin} be feasible at time $t$. Let us denote the corresponding optimal input policies as $[\pi^\star_{t|t}(\cdot),\pi^\star_{t+1|t}(\cdot),\cdots,\pi^\star_{t+N-1|t}(\cdot)]$. Assume the MPC controller $\pi^\star_{t|t}(\cdot)$ is applied to \eqref{eq:unc_sys} in closed-loop and $E^d(\mathcal{X}_{k|t+1})$ for $k=\{t+1,t+2,\dots,t+N+1\}$ are obtained according to \eqref{eq:bigE},~\eqref{eq:uncer_predBound} and \eqref{eq:pred_succ}. Consider a candidate policy sequence at the next time instant as:
\begin{align}\label{eq:feas_seq_next_DF_sto}
    \Pi_{t+1}(\cdot) = [\pi^\star_{t+1|t}(\cdot),\dots,\pi^\star_{t+N-1|t}(\cdot),Kx_{t+N|t+1}].
\end{align}
From \eqref{eq:prev_unc} and Proposition~\ref{prop:containment} we conclude that the policy sequence $[\pi^\star_{t+1|t}(\cdot),\pi^\star_{t+2|t}(\cdot),\dots,\pi^\star_{t+N-1|t}(\cdot)]$ is an $(N-1)$ step feasible policy sequence at $t+1$ (excluding terminal condition), since at previous time $t$, it robustly satisfied all stage constraints in \eqref{eq:MPC_R_fin}. With this feasible policy sequence, $x_{t+N|t+1} \in \mathcal{X}_N$. From \eqref{eq:term_set_DF} we conclude that \eqref{eq:feas_seq_next_DF_sto} ensures $x_{t+N+1|t+1} \in \mathcal{X}_N$. This concludes the proof. 
\end{proof}

\section{Numerical Example}\label{sec:numerics}

In this section we demonstrate both the aspects of exploration and robust control of our robust Adaptive MPC, highlighted in Algorithm~\ref{alg1}. We wish to compute feasible solutions to the following infinite horizon control problem
\begin{equation}\label{eq:generalized_InfOCP_ex}
	\begin{array}{llll}
		\hspace{0cm}    
			\hspace{0cm} 
	\displaystyle\min_{u_0,u_1(\cdot),\ldots} & \displaystyle\sum\limits_{t\geq0} \bar{x}_t^\top Q \bar{x}_t + u_t^\top(\bar{x}_t) R u_t(\bar{x}_t)  \\[1ex]
	\hspace{5mm}	\text{s.t.} & x_{t+1} = Ax_t + Bu_t(x_t) + 0.05 \begin{bmatrix}\tan^{-1}(x_t(1)) \\ x_t(2)\end{bmatrix}\\[2.5ex]
		& \begin{bmatrix}-1 \\ -1 \\ -4
		\end{bmatrix} \leq \begin{bmatrix}x_t \\ u_t(x_t)
		\end{bmatrix} \leq \begin{bmatrix}1 \\ 3 \\ 1
		\end{bmatrix},~(\mathcal{X} \times \mathcal{U})\\[3.5ex]
		& \forall d(x_t) \in \mathcal{D}(x_t), \\
		&  x_0 = \bar{x}_0= x_S,\ t=0,1,\ldots,
	\end{array}
\end{equation}
with initial state $x_S = [-1,2]^\top$, where $ A = \begin{bmatrix} 1.2 & 1.5\\
		0 & 1.3\end{bmatrix}$ and $B = [0,1]^\top$. Algorithm~\ref{alg1} is implemented with a control horizon of $N=3$, and the feedback gain $K$ in \eqref{eq:term_set_DF} is chosen to be the optimal LQR gain for system $x^+ = (A+BK)x$ with $Q = 10 \mathbb{I}_2$ and $R= 2$. 

\subsection{Exploration for Uncertainty Learning}
We initialize $\mathcal{P}^d(\mathcal{X}) = \R^n$, resulting in an empty terminal set $\mathcal{X}_N$ in \eqref{eq:MPC_R_fin}. In this section, we present the ability of Algorithm~\ref{alg1} to explore the state-space with randomly generated inputs $u_j \sim \mathcal{N}(0,1)$, in order to eventually obtain a nonempty  $\mathcal{X}_N$ for starting the control process. 
\begin{figure}[h]
	\centering
	\includegraphics[width=\columnwidth]{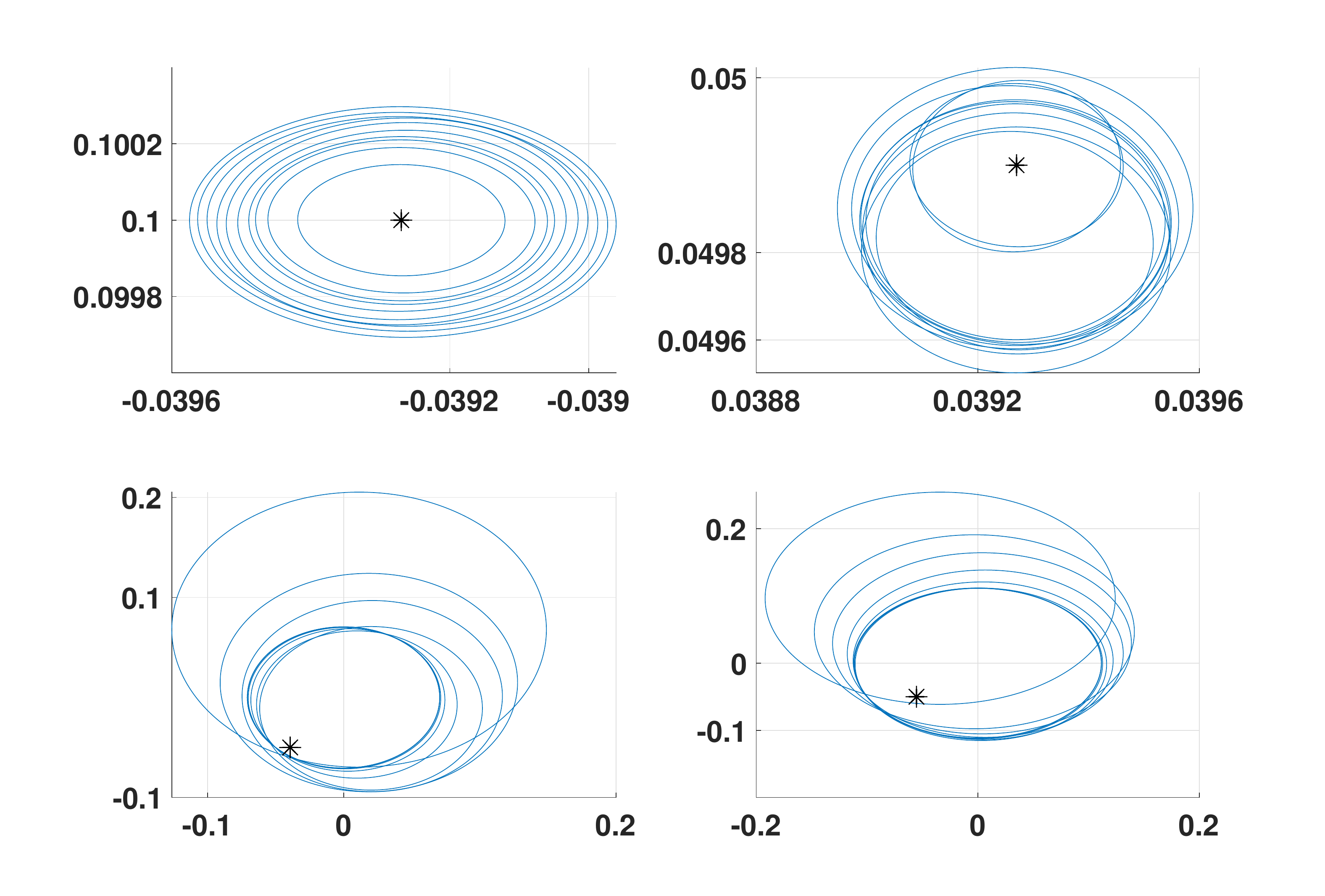}
	\caption{Uncertainty bound $\mathcal{D}(x)$ estimation at \emph{query} points with successive intersection of ellipses obtained from measured data. Star ($\star$) denotes the true value of $d(x)$, lying in the intersection.}
	\label{fig:states_explore}    
\end{figure}
Let the time indices during exploration phase be denoted by $j$. 

Fig.~\ref{fig:states_explore} shows the sets $E^d({x})$ at four fixed \emph{query} points ${x}_j = \{[-1,2]^\top,[1,1]^\top,[-1,1]^\top,[-2,-1]^\top \}$ as data is collected until instant $j$. This can be obtained from feasibility of \eqref{eq:bigE} (with $k = j$). As $j$ increases, $E^d(x)$ for each $x$ is contained in the successive intersections of ellipsoids, from \eqref{eq:func_dom}. The intersection shrinks for all points, as claimed in Proposition~\ref{prop:containment}. This is seen in Fig.~\ref{fig:states_explore}, which indicates improved information of $\mathcal{D}(x)$ with added data, for all $x \in \mathcal{X}$. At $j_\mathrm{max} = 30$, a nonempty $\mathcal{X}_N$ is obtained, shown in Fig.~\ref{fig:term_set}. This is when we start control and set $t=0$. 

\begin{figure}[h]
	\centering
	\includegraphics[width=\columnwidth]{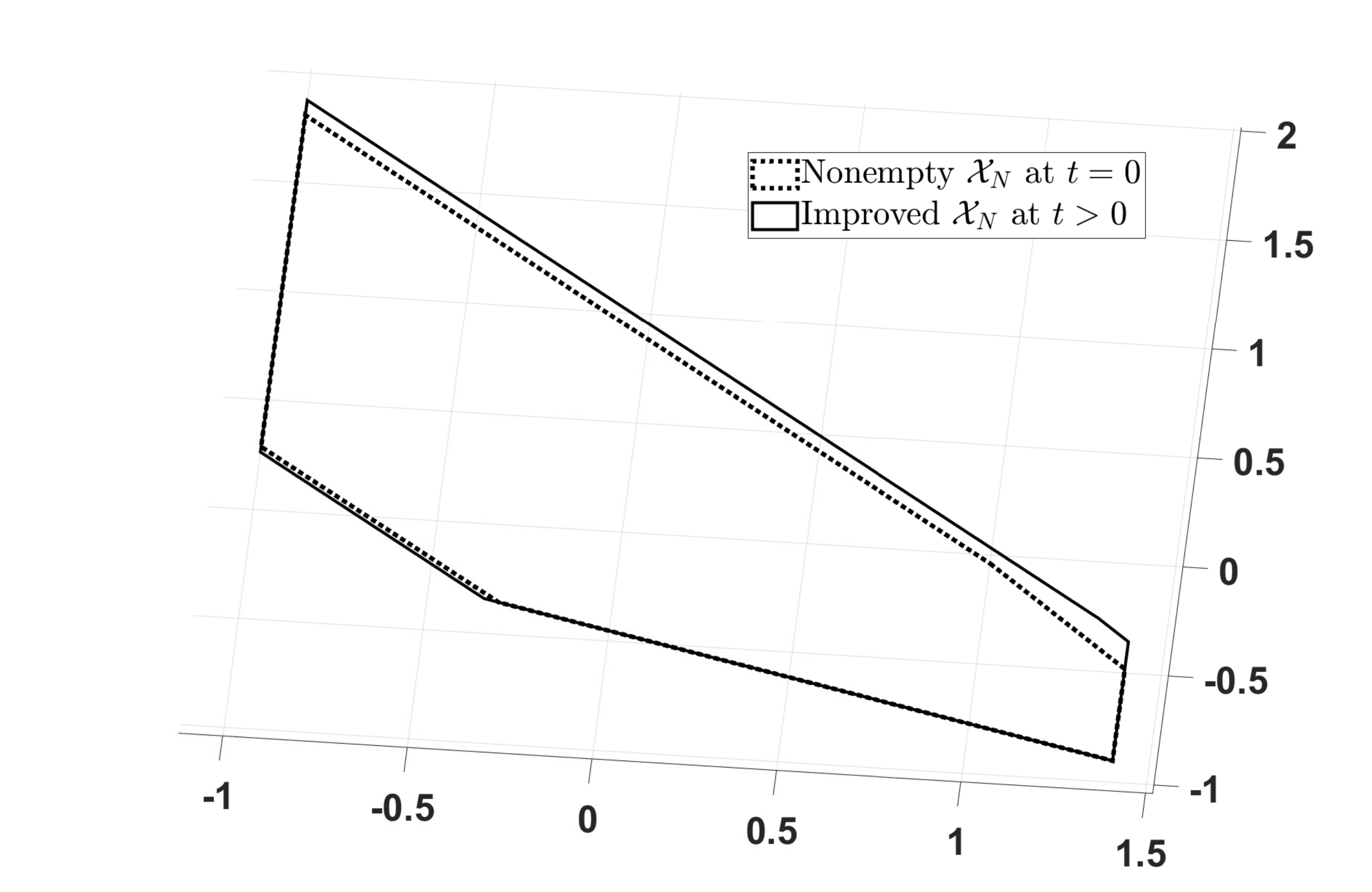}
	\caption{Terminal set construction. The set grows as estimation of $d(x)$ is improved from measurements.}
	\label{fig:term_set}    
\end{figure}

\subsection{Robust Constraint Satisfaction}
If the MPC problem \eqref{eq:MPC_R_fin} is feasible for parameters defined in \eqref{eq:generalized_InfOCP_ex}, it ensures robust satisfaction of constraints in \eqref{eq:generalized_InfOCP_ex} for all times $t>0$. This is highlighted with a realized trajectory in Fig.~\ref{fig:state_traj}. Furthermore, the terminal set is recomputed and improved at a $t>0$ with \eqref{eq:term_set_DF}, having refined $\mathcal{P}^d(\mathcal{X})$ estimation\footnote{rectangles with sides of length equal to major and minor axes of $E^d(\cdot)$} from \eqref{eq:bigE} (with $k=t+N$). The set grows, as seen in Fig.~\ref{fig:term_set}, resulting in lesser conservatism of \eqref{eq:MPC_R_fin}.

\begin{figure}[ht]
	\centering
	\includegraphics[width=\columnwidth]{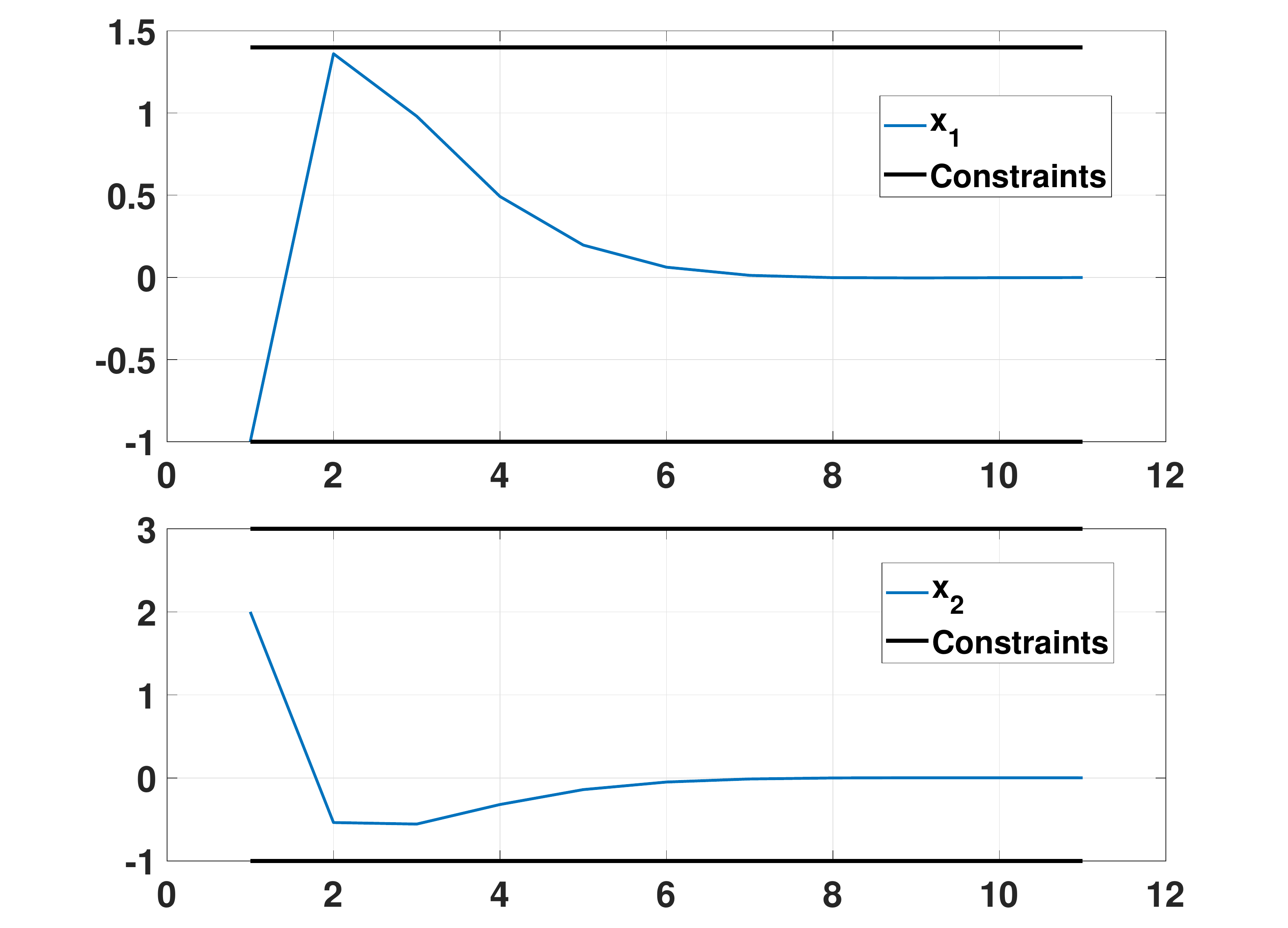}
	\caption{State trajectory with robust constraint satisfaction.}
	\label{fig:state_traj}    
\end{figure}


\section{Conclusions}
We proposed an Adaptive MPC framework to achieve robust satisfaction of state and input constraints for uncertain linear systems. The system uncertainty is assumed state dependent and globally Lipschitz. An envelope containing the uncertainty range is constructed with Quadratic Constraints (QCs), and is refined with data as the system explores the state-space. Upon collection of sufficient data, the system is able to solve a robust MPC problem for all times from a given initial state. The algorithm further reduces its conservatism by incorporating online model adaptation during control. 

\section*{Acknowledgements}
This work was partially funded by Office of Naval Research grant ONR-N00014-18-1-2833.
\renewcommand{\baselinestretch}{0.95}
\bibliographystyle{IEEEtran}
\bibliography{IEEEabrv,bibliography} 

\section*{Appendix}
\subsection*{Proof of Proposition~\ref{prop:containment}}
Let $x_i$ be the measurements collected at any time instant $i < t$. From \eqref{eq:func_dom} we see that for any given time $t$, the uncertainty domain $\mathcal{D}(\tilde{x}_t)$ is obtained from successive intersection operations of sampled range sets at $\tilde{x}_t$, for all times until $t$. Hence, $\mathcal{D}(\tilde{x}_{t_2}) = (\bigcap_{i=0}^{t_1} \mathcal{S}(x_i, \tilde{x}_{t_1})) \bigcap_{i=t_1}^{t_2} \mathcal{S}(x_i, \tilde{x}_{t_2}) = \mathcal{D}(\tilde{x}_{t_1}) \bigcap_{i=t_1}^{t_2} \mathcal{S}(x_i, \tilde{x}_{t_2})$, implying $\mathcal{D}(\tilde{x}_{t_2}) \subseteq \mathcal{D}(\tilde{x}_{t_1})$.

\subsection*{Proof of Proposition~\ref{prop:containment2}}
Consider any vector $[x^\top d^\top 1]^\top \in\mathbb{R}^{2n+1}$ such that $x\in E^d(\mathcal{X}_{k|t})$ and $[x^\top d^\top]^\top \in G(d)$. Given that \eqref{eq:bigE} is feasible for each prediction instant $k = \{t+1,\dots,t+N\}$ at any time $t$, we multiply $[x^\top d^\top 1]^\top$ on both sides of \eqref{eq:bigE}
\begin{align*}
    & -\rho^k_t \begin{bmatrix}x\\1\end{bmatrix}^\top \bar{P}^x_{k|t} \begin{bmatrix}x\\1\end{bmatrix}  
    + \begin{bmatrix} d \\ 1 \end{bmatrix}^\top \bar{P}^d_{k|t} \begin{bmatrix}d\\1\end{bmatrix}\\
    &~~~~~~~~~~~~~~~~~~~~~~~~~~~~~~-\begin{bmatrix}x \\ d \\ 1 \end{bmatrix}^\top \sum \limits_{i=0}^{t-1} \tau^k_i Q_L^d(x_i)\begin{bmatrix}x\\d\\1\end{bmatrix}\leq 0,
\end{align*}
for some $\{\rho^k_t, \tau^k_0,\dots,\tau^k_{t-1}\} \geq 0$, where $\bar{P}^x_{k|t}$ and  $\bar{P}^d_{k|t}$ are defined in Section~\ref{ssec:bounding_unc_set}. Now using Corollary~\ref{corr:env}, \eqref{eq:X_setDef} and \eqref{eq:prev_unc}, we can infer $\begin{bmatrix}d\\1\end{bmatrix}^\top \bar{P}^d_{k|t} \begin{bmatrix}d\\1\end{bmatrix}\leq 0$. 

\subsection*{SDP for Solving \eqref{eq:bigE}}
For all $k = \{t+1,\dots,t+N\}$, along MPC horizon, let us use the variable nomenclature $p(\mathcal{X}_{k|t}) = - \rho^k_t q^x_{k|t} + \sum \limits_{i=0}^{t-1} \tau^k_i L_d^2 \mathbb{I}_n,~q(\mathcal{X}_{k|t}) = \rho^k_t (q^x_{k|t})^\top p^x_{k|t} - \sum \limits_{i=0}^{t-1} \tau^k_i L_d^2x_i,~r(\mathcal{X}_{k|t}) = -\sum \limits_{i=0}^{t-1} \tau^k_i \mathbb{I}_n,~s(\mathcal{X}_{k|t}) =  \sum \limits_{i=0}^{t-1} \tau^k_i d(x_i)$, and $t(\mathcal{X}_{k|t}) = \rho^k_t \Big ( 1-(p^x_{k|t})^\top q^x_{k|t} p^x_{k|t} \Big ) - \sum \limits_{i=0}^{t-1} \tau^k_i \Big ( -L_d^2x_i^\top x_i + d^\top(x_i) d(x_i) \Big ) - 1$. Finding the minimum trace ellipsoid satisfying \eqref{eq:bigE} is posed as an SDP \cite[Section~11.4]{calafiore2014optimization} as:
\begin{equation*} 
    \begin{array}{llll}
        \displaystyle \min_{\xi} & \textnormal{trace}((q^d_{k|t})^{-1})\\
        \ \  \text{s.t.} 
        & \begin{bmatrix}  p(\mathcal{X}_{k|t}) & \mathbf{0} & q(\mathcal{X}_{k|t}) & \mathbf{0} \\ \mathbf{0} & r(\mathcal{X}_{k|t}) & s(\mathcal{X}_{k|t})& -\mathbb{I}_n \\ q^\top(\mathcal{X}_{k|t}) & s^\top(\mathcal{X}_{k|t}) & t(\mathcal{X}_{k|t}) & (p^d_{k|t})^\top \\ \mathbf{0} & -\mathbb{I}_n & p^d_{k|t} & -(q^d_{k|t})^{-1} \end{bmatrix} \preceq 0, \\ 
        & \rho^k_t \geq 0,\tau^k_i \geq 0,~q^d_{k|t} \succ 0, \\
        &\forall i = 0,1,\dots, t-1,
    \end{array}
\end{equation*}
with $ \xi = \{q^d_{k|t}, p^d_{k|t}, \rho^k_t, \tau^k_0,\dots,\tau^k_{t-1}\}$ and $\mathbf{0} \in \R^{n \times n}$. 

\end{document}